\documentclass[runningheads]{llncs}

\pdfoutput=1

\usepackage{amsmath}
\usepackage{amsfonts}
\usepackage{amssymb}
\usepackage[clock]{ifsym}
\usepackage{caption}
\usepackage{todonotes}

\usepackage{paralist}
\usepackage{graphicx}    

\usepackage{tikz}
\usetikzlibrary{shapes,arrows,calc,shapes.gates.logic.US,shapes.gates.logic.IEC}
\usetikzlibrary{circuits}
\usetikzlibrary{patterns,decorations.pathreplacing}

\usepackage[capitalize]{cleveref}
\usepackage{comment}

\newcommand{\N}{\mathbb{N}}
\newcommand{\R}{\mathbb{R}}

\newcommand{\BO}{\mathcal{O}}

\newcommand{\In}{\operatorname{\mbox{\sc in}}}
\newcommand{\Out}{\operatorname{\mbox{\sc out}}}
\newcommand{\En}{\operatorname{\mbox{\sc en}}}
\newcommand{\Ysig}{\operatorname{\mbox{\sc y}}}

\newcommand{\Asig}{\operatorname{\mbox{\sc a}}}
\newcommand{\Bsig}{\operatorname{\mbox{\sc b}}}

\newcommand{\E}{{E}}
\newcommand{\ME}{\mathcal{E}}
\newcommand{\WM}{\mbox{WM}}
\newcommand{\eSPF}{\mbox{eSPF}}

\begin{document}

\title{On Specifications and Proofs of Timed Circuits\thanks{%
This project has received funding from the European Research Council 
  (ERC) under the European Union's Horizon 2020 research and innovation programme (grant agreement 716562).
The work was supported by the Austrian Science Fund project DMAC (P32431) and
  by the ANR project DREAMY (ANR-21-CE48-0003).}}

\author{
Matthias F\"ugger \inst{1} \orcidID{0000-0001-5765-0301} \and 
Christoph Lenzen \inst{2} \orcidID{0000-0002-3290-0674} \and
Ulrich Schmid \inst{3}\thanks{Corresponding author (\emph{s@ecs.tuwien.ac.at}).} \orcidID{0000-0001-9831-8583}}

\authorrunning{M. F\"ugger, C. Lenzen, and U. Schmid}

\institute{
CNRS, LMF, ENS Paris-Saclay, Universit\'e  Paris-Saclay, Inria \and
CISPA Saarbr\"ucken \and
TU Wien}

\maketitle

\begin{abstract}
Given a discrete-state continuous-time reactive system, like a digital circuit,
  the classical approach is to first model it as a state transition system and then prove
  its properties.
Our contribution advocates a different approach: to directly operate on the input-output behavior
  of such systems, without identifying states and their transitions in the first place.
We discuss the benefits of this approach at hand of some examples, which demonstrate that it nicely integrates with concepts of self-stabilization and fault-tolerance. 
We also elaborate on some unexpected artefacts of module composition in our framework, and
conclude with some open research questions.
\end{abstract}

\section{Motivation and Overview}

Many physical systems dealt with by computational methods today do not operate on
  discrete values.
Examples range from electronic circuits to mechanical systems to chemical processes,
which all share that analog information is continuously processed.
Whereas natural and engineering sciences, in particular,
control theory, have been successful in finding and using accurate
  models for such continuous systems, primarily based on systems of differential equations,
  the complexity both involved in model development and model usage is often prohibitive:
Model composition and hierarchical modeling is usually difficult, and large simulation
  times and memory consumption as well as numerical instability typically limit the
  applicability of the resulting models in practice.

\noindent\textit{Discrete-valued abstractions.}
  Applying discrete-valued abstractions for modeling continuous-valued systems is hence
  an attractive alternative, and much of the big success of computer science is owed
  to their introduction.
Apart from being easily specified and understood, discrete abstractions usually
  involve all-digital information and finite (typically small) state-spaces that can be
  efficiently processed, transmitted, and stored.
Among the many success stories of this approach is digital circuit design,
  which is one of the key enablers of modern computer systems (and also
our main source of application examples):
While it is out of question to perform analog simulations of the billions of transistors and
  other analog electronic components that implement the logic gates in a modern
  \emph{very-large scale integration} (VLSI) circuit, applying digital timing
  simulations and verification techniques is common practice.

\noindent\textit{The need for accurate timed circuit models.}
Given the tremendous advances in VLSI technology, with clock speeds in
  the GHz range and voltage swings well below 1\,V \cite{ITRS07},
  modeling accuracy becomes a concern \cite{Con03}.
  Additionally, manufacturing \emph{process, temperature, and supply-voltage}
  (PVT) variations cause a large
  variability in switching and signal propagation delays.
Furthermore, reduced critical charges make the circuits more susceptible to ionizing
  particles \cite{Bau05,DW11} and
  electromagnetic interference \cite{MA01}, and feature sizes in the 10\,nm
  range also increase the likelihood of permanent errors due to manufacturing
  defects and wear-out \cite{KK98,PB93}.
All these effects together make non-conservative delay predictions, which are required for 
digital modeling of fast synchronous
  circuits, difficult to obtain.

Indeed, none of these effects is adequately captured by existing timed digital
  circuit models.
Besides the lack of modeling and analysis support for
  fault-tolerance, it was shown in \cite{FNS16:ToC} that none of the classic
  digital channel models, including the widely used pure and inertial delay channels
  \cite{Ung71}, faithfully model the
  propagation of short pulses.
The same is true for more advanced models like PID channels \cite{BJAV00} (with the notable
  exception of involution channels \cite{FNNS19:TCAD}, though).
Since existing digital simulators exclusively use
  classic delay models, their predictions are hence not always accurate.
Moreover, existing digital design tools lack an adequate support for
  metastability\footnote{Metastable upsets can occur
  in any state-holding device with discrete stable states, such as memory cells.
  If a new state transition is triggered before the state change caused by the
  previous one has settled, an intermediate output value
  may be observed arbitrarily late.
  Even Byzantine (i.e., ``worst-case'') fault-tolerance techniques are incapable of
  containing the effects of metastable upsets perfectly \cite{FFL18:ToC}, since a signal
  outside the (discrete) value domain is not just an arbitrary regular signal.
  Metastability is not restricted to electrical systems.
  For example, an engineered genetic toggle switch~\cite{GCC00},
    acting as a memory cell storing $0$ or~$1$, was observed to 
    exhibit metastable behavior besides its two stable states.}
  \cite{Mar81} modeling and analysis.

\noindent\textit{Modeling approaches: state-based and state-oblivious.}
A natural and powerful tool for modeling such systems are transition systems.
A transition system is defined by a set of states, transitions between
  these states, and rules how \emph{executions}, i.e., (timed) state sequences,
  are generated by such a system.
Transition systems can be white-box or black-box: white-box approaches try to
  follow the actual implementation and model the dynamics of a system's state,
  while black-box models just try to capture the dynamics of the system's
  inputs and outputs.
In the latter case, states are merely used as equivalence classes of
  execution prefixes (histories), to abstract away individual execution prefixes that do not
  need to be further distinguished when capturing the system's
  behavior.
We will refer to both variants as \emph{state-based specifications} in the following.

On the other hand, the correct behavior of a system can be
  directly specified by the set of valid executions.
For convenience, this is typically done in terms of an input-output
  function that maps a (timed) input state sequence to a set of
  allowed (timed) output state sequences.
We will refer to this as a \emph{state-oblivious specification}.

While the distinction is not strict, as a state-oblivious specification is easily translated
  to a state-based specification with the set of states being the
  set of all execution prefixes, i.e., trivial equivalence classes taken as states,
  the approaches tend to lead to quite different formalizations and proofs.
In computer science, state-based specifications have received lots of attention,
and can nowadays draw from a rich set of techniques, e.g., for showing that one system implements another system via
  simulation relations, or reasoning about properties of compositions of systems.

\noindent\textit{A take on a state-oblivious modeling framework.}
In this work, we advocate the considerably less popular alternative of
state-oblivious
  specifications and discuss some of its properties, building on the
framework originally presented in \cite{dolev14}. 

In \cref{sec:time} to \cref{sec:transientfaults}, we review
  the cornerstones of state-oblivious formalizations of continuous-time,
  dis\-crete-valued circuits as introduced in \cite{dolev14}.
Like in some existing approaches for reactive systems, such as Broy and St{\o}len's
  FOCUS \cite{BS01:focus}, a \emph{module} is specified directly in terms of the
  output behaviors that it may exhibit in response to a given input signal.
As already said, this is very different from existing frameworks that follow a state-based approach,
  including Alur-Dill Timed Automata
  \cite{Alur1994}, Lamport's TLA \cite{Lam94:TLA}, Timed IO Automatons by Keynar
  et.~al.~\cite{kaynar06}, and discrete abstractions for hybrid systems
  \cite{AHLP00}, as well as state-space-based control theory \cite{LV11}, which
  all act on the (sometimes uncountable) state-space of the underlying system.

We demonstrate that typical timing constraints and fault-tolerance properties (e.g.\ Byzantine behavior \cite{PSL80}) are easily expressed and dealt with in a
  state-oblivious framework.
In \cref{sec:SSOsc}, we also demonstrate that self-stabilizing systems \cite{Dol00}, in which the initial internal state
  is assumed to be completely arbitrary after a catastrophic (but transient) event,
  can be described and proved correct appropriately.

In \cref{sec:composition}, we address composition of modules.
Composition has been extensively studied in state-based approaches. In general, 
however, it is even difficult to decide whether behavioral specifications match
  at interface boundaries \cite{AHS02}.
While one can prove some generic properties about composition of state-oblivious specifications, they apply to quite restricted settings only.
In \cref{sec:weirdmodule}, we will show that there are indeed some unexpected module composition artefacts for state-oblivious specifications when one
considers more general settings.
In particular, the \emph{eventual short pulse filter} (eSPF) module introduced in 
  \cite[Sec.~7]{FNS16:ToC} reveals that composing modules with bounded delay 
  in a feedback loop may result in a module with finite delay.
Even worse, whereas the feedback-free composition of bounded-delay modules is always bounded delay,
  it turns out that the \emph{feedback-free} composition of finite-delay modules need not always be finite-delay. 
Some conclusions and problems open for future research are provided in
  Section~\ref{sec:outlook}.

\noindent\textit{Applications beyond VLSI circuits.}
Although our framework emerged in the context of digital
  circuits~\cite{dolev14}, its applicability extends to
  other domains. It needs to be stressed, though, that
  we do not aim at typical application domains of classic
  control theory. Despite some similarities, like considering
  systems as function transformers, our continuous-time discrete-value
  (typically binary) signals and their properties of interest
  are very different from the signals considered
  in either continuous or discrete control theory.

  However, the idea to model continuous dynamical
  systems by discrete-state timed circuits has been successfully applied to
  genetics as well:
Rather than analyzing dependencies of transcription and
  protein levels by means of differential equations, genetic circuit models have been
  used for descriptive~\cite{AO03,Tho73} and synthetic~\cite{GCC00,HMC02,SCBMTH08,BBN13}
  purposes. In the meantime, a body of genetic circuit design principles
  has been established \cite{nielsen2016genetic,gorochowski2017genetic,santos2020multistable}.
For a discussion on differences between classical circuits in silicon and genetic circuits,
  we refer the reader to \cite{fugger2020digital}. 
Further, as many biological systems are fault-tolerant and even
  self-stabilizing to a certain extent, a unified model bears the promise of
  cross-fertilization between different application domains.
Earlier work on biologically inspired self-stabilizing Byzantine fault-tolerant clock
  synchronization \cite{DDP03:SSS,sivan99lobster} is a promising example of the
  benefit of this approach.

\section{Timed Circuit Models}
\label{sec:time}

Before discussing basic properties of state-oblivious specifications via
  some examples, we very briefly recall the standard synchronous, asynchronous, and
  partially synchronous timing models for specifying distributed systems.
Obviously, they can be also used for modeling gates in a circuit that 
communicate with each other via interconnecting wires.

In \emph{synchronous systems}, components act in synchronized lock-step
\emph{rounds}, each comprising a communication phase and a single
computing step of every component. The strict constraints on the order
of computing steps thus facilitate algorithms that are
simple to analyze and implement, yet can leverage time to, e.g.,
avoid race conditions and implement communication by time~\cite{lamport78}.
Unfortunately, implementing the
synchronous abstraction, e.g., by central clocking or causal relations
enforced by explicit communication~\cite{awerbuch85}, can be too
inefficient or plainly infeasible.

This fact fuels the interest in \emph{asynchronous systems}, for which no assumptions
are made on the order in which computation and communication steps are executed.
The standard way of modeling asynchronous executions is to associate a local
  state with each component, and let a \emph{(fair) scheduler} decide in which order
  components communicate and update their states (i.e., receive information and
  perform computation).
Viewing synchrony as the temporally most ordered execution model of
  distributed computations, asynchronous systems are at the other extreme end of
  the spectrum. Since it is impossible to distinguish very slow components
  from such that have suffered from a crash fault, however, proving correct
  asynchronous distributed algorithms is difficult and often impossible.\footnote{Consensus~\cite{PSL80}, a basic fault-tolerance
  task, can be solved deterministically in synchronous systems, but not in
  asynchronous systems~\cite{fischer85impossibility}.}

To circumvent this problem, a number of intermediate \emph{partially-synchronous} state-based models have been
  defined (e.g.~\cite{Dolev:1987,DLS88}).
However, such models typically serve either
  special applications or as a vehicle to better understand the fundamental
  differences between synchronous and asynchronous systems.

\subsection{When state-based formalizations are unnecessarily complicated}
\label{sec:vs-statebased}

While synchronous and asynchronous systems are easy to specify within
  state-based frameworks, the situation becomes different for systems with
  more complicated timing constraints like partially synchronous systems.
The challenge in allowing for general timing constraints is that the elegant and
  convenient separation of time and the evolution of the system state cannot be
  maintained.
The situation becomes even more involved when the goal is to model circuits, as
  opposed to software-based computer systems.
A major difference is that software-based systems typically reside at a level of
  abstraction where discrete, well-separated actions are taken naturally by an underlying machine.
The evolution of the internal state of this machine is then modeled as a
  transition system.
  By contrast, real circuits are analog devices that continuously transform
  inputs into outputs.

We demonstrate the differences between the state-based and the state-oblivi\-ous
approach at the example of the arguably simplest circuit, namely, a
bounded-delay channel, as instantiated e.g.\ by an (ideal) wire.

\noindent\textit{A channel.}
We consider a binary bounded-delay first-in first-out (FIFO) channel, which
has a single input port (= connector) and a single output port.
Whereas such channels are also employed in various state-based models, 
they are usually \emph{part} of the model and typically also the only
  means of communication.
By contrast, we describe the channel as an object in a (to-be-defined) model.

Informally, we require the following:
The \emph{input port} is fed by an input signal given as
a function $\In: \R \to \{0,1\}$. Note that the restriction
to binary-valued  signals is for simplicity only and could be replaced
by arbitrary discrete ranges. The reason why the domain of $\In$
is $\R$ instead of, e.g., the non-negative reals $\R^+_0$, will be
explained later.  
Typically, some further restrictions are made on (input) signals for
modeling real circuits, e.g.,
  only a finite number of transitions within each finite interval.
For each input signal, the \emph{output port} produces an output signal such that:
  (i) for each input transition there is exactly one output transition within some time $d > 0$,
  and (ii) output transitions occur in the same temporal order as their corresponding input transitions.

\noindent\textit{The state-based approach: the channel as a transition system.}
A state-based description would model the state of the channel at 
some time~$t$,
  as well as the rules for transitioning between states.
Obviously, this would allow us to infer the correct behavior of the channel
  at times greater than~$t$ as valid traces of this transition system.
However, a state-based description would be at odds with our goal of a simple and modular
  specification of the system:
\begin{compactitem}
  \item Both a physical wire connecting sender and receiver and reliable 
    multi-hop wireless channel are implementations of our channel.
  Specifications with one of them in mind may differ significantly.

  \item The state space of the channel would be infinitely large, as it must be able to 
    record an arbitrarily long sequence of alternating
    input transitions that may have occurred within $[t-d,t]$

  \item The strategy of breaking down a difficult-to-describe 
state-based description into
    smaller building blocks does not help here.
\end{compactitem}

The above problems become even more pronounced when it comes to more interesting
  modules.
Even if we were not discouraged by the above difficulties and went for a state-based definition of
  the channel, e.g. in terms of Timed I/O Automata~\cite{kaynar06}, we argue that the original advantage of a
  state-based approach would be lost:
  the canonical description of the global state of the system as the product of
  the components' states.

\noindent\textit{The state-oblivious approach: the channel as an input-output function.}
We conclude that our preferred option is (i) to treat the channel as a blackbox, and
  (ii) not to bother with finding states, i.e., equivalence classes of histories, in the
  first place.
That is, we infer the possible output not from some internal state, but rather directly
  from the input history.
By the nature of a channel, however, the feasible output values at time~$t$ cannot be
  determined from $\In$ alone: the \emph{output} at times smaller than~$t$ enters the bargain as well.
Hence, we naturally end up directly relating input and output signals:
For each (possible) input signal $\In$, there must be a non-empty \emph{set} of
  \emph{feasible} output functions $\phi(\In)$, where the
  \emph{module specification} $\phi$ maps inputs signals to such feasible output signals. Note that we allow the adversary to choose which of the feasible output signals a
  module generates in some execution, i.e., we just assume non-determinism
  here.\footnote{Whereas one could extend our framework to restrict the adversary here,
  e.g., to capture probabilistic choice, we will not consider this possibility
  in this work.}
  This also allows to express any given restriction on the inputs, e.g., one that is considered suitable for a given module, simply by permitting
  \emph{any} output signal for input signals that violate such a restriction.

A state-oblivious specification of our channel can be 
given in terms of an input-output
  function~$\phi$.
For every input signal $\In$, the output signal $\Out$ is
  feasible, i.e., $\Out \in \phi(\In)$ if
%\begin{align*}
  $\Out(t)=\In\left(\delta^{-1}(t)\right)$,
%\end{align*}
  where the delay function $\delta: \R \to \R$ is continuous, strictly
  increasing (hence invertible), and satisfies
%\begin{align*}
  $t \leq \delta(t) \leq t+d$
%\end{align*}
  for all $t\in \R$.

%\medskip

As we model signals as functions of
  real-time, we do not need special signal values (as in FOCUS \cite{BS01:focus}) 
  that report the progress of time, and relating different (input, output) signals to each
  other becomes simple.
  
It remains to explain why we chose the whole set of reals $\R$ as the (time) domain of
  our input and output functions.
Again, in principle, nothing prevents us from using functions $\R^+_0\to \{0,1\}$.
For the considered channel, it would be reasonably easy to adapt the description:
  $\Out\in \phi(\In)$ is arbitrary on $[0,\delta(0))$ and
  $\Out(t)=\In(\delta^{-1}(t))$ on $[\delta(0),\infty)$, for some continuous,
  strictly increasing delay function $\delta:\R^+_0\to\R^+_0$ satisfying that
  $t\leq \delta(t)\leq t+d$ for all $t\in \R^+_0$.

However, given that our ``equivalent'' to the channel's state is its input history,
  it is more natural to rely on input signals with a time domain that contains
   $[-d,\infty)$ when specifying feasible outputs on $\R^+_0$.
Extending the range by a finite value only would not cover all possible
  values of $d$, though.
Moreover, there are modules whose output may depend on
  events that lie arbitrarily far in the past, e.g., a memory cell.
For simplicity and composability, picking $\R$ as domain is thus preferred here.
We remark that this convention does not prevent suitable initialization of a
  module, say at time $t=0$, however.

\section{Composition}
\label{sec:composition}

In the previous section, we demonstrated the use of state-oblivious specifications
  in terms of directly providing input-output functions $\phi$ for a simple channel. 
In general, we define:

\begin{definition}[Module]
  \label{def:module}
  A \emph{signal} is a function from $\R$ to $\{0,1\}$.
  A \emph{module} $M$ has a set of \emph{input ports} $I(M)$ and
  a set of \emph{output ports} $O(M)$, which are the connectors where
  input signals are supplied to $M$ and output signals leave $M$.
  The \emph{module specification} $\phi_M$ maps the input signals 
$(\In_p:\R\to\{0,1\})_{p\in I(M)}$ to sets of allowed output signals
$(\Out_p:\R\to\{0,1\})_{p\in O(M)}$.
An \emph{execution} of a module is a member of the set 
\[
\left\{\bigl( (\In_p)_{p\in I(M)},
                 \phi_M((\In_p)_{p\in I(M)}) \bigr) \mid (\In_p:\R\to\{0,1\})_{p\in I(M)} \right\}\,.
\]
\end{definition}
Note that we typically assume that modules are \emph{causal},
i.e., that images of $\phi_M$
  for two inputs that are identical until time $t$, are identical until time $t$.

Specifying a module~$M$ this way, i.e., by providing $\phi_M$, can either be viewed as
  stating an assumption, in the sense that it is already known how to build a module with the
  respective behavior, or as stating a problem, i.e., expressing a desired behavior of a module
  that still needs to be built.
We call a module specified this way a \emph{basic module}.

\emph{Implementing} such a module can be done in two different ways:
  (i) directly within a target-technology, which leaves the scope of our
modeling framework, or
  (ii) by decomposition into smaller modules within the modeling framework.

Let us now formalize what the latter means in the context of our approach.
Intuitively, we will take a set of modules and connect their input and output
  ports to form a larger \emph{compound module}, whose inputs and outputs are subsets
  of the ports of these modules (cp.~Figure~\ref{fig:osc}).
The input-output function of the compound module is then derived from the ones of the
  submodules and their interconnection.

\begin{definition}[Compound module]
\label{def:comp_module}
A compound module~$M$ is defined by:
\begin{compactenum}
\item Decide on the sets of \emph{input ports} $I(M)$ and
\emph{output ports} $O(M)$ of $M$.
\item Pick the set $S_M$ of \emph{submodules} of $M$. Each submodule $S\in
S_M$ has input ports $I(S)$, output ports $O(S)$, and
a specification $\phi_S$ that maps tuples
$(\In_p:\R\to\{0,1\})_{p\in I(S)}$ of functions to sets of tuples
$(\Out_p:\R\to\{0,1\})_{p\in O(S)}$ of functions, satisfying the following
well-formedness constraints:
\begin{compactitem}
\item For each output port $p\in O(M)$, there is exactly one submodule $S\in
S_M$ such that $p\in O(S)$.
\item For each input port $p\in I(S)$ of some submodule $S\in S_M$, either
$p\in I(M)$ or there is exactly one submodule $S'\in S_M$ so that $p\in O(S')$.
\end{compactitem}
\item For each $(\In_p:\R\to\{0,1\})_{p\in I(M)}$, we require
%\begin{align*}
$\left( \Out_p:\R\to\{0,1\} \right)_{p\in O(M)} \in
\phi_M\left( (\In_p)_{p\in I(M)} \right)$
%\end{align*}
iff there exist functions
$(f_p:\R\to\{0,1\})_{p\in \bigcup_{S\in S_M} I(S)\cup O(S)}$ with
\begin{compactitem}
\item $\forall S\in S_M: (f_p)_{p\in O(S)}\in \phi_S((f_p)_{p\in I(S)})$;
\item $\forall p\in I(M): f_p=\In_p$; and
\item $\forall p\in O(M): f_p=\Out_p$.
\end{compactitem}
\end{compactenum}
\end{definition}
Note that choices are made only in Steps~1 and~2, whereas $\phi_M$ is defined
implicitly and non-constructively in Step~3. Informally, the latter just
says that any execution, i.e., any pair of input and output signals,
of $M$ that leads to feasible executions of all submodules, must be in
$\phi_M$.

\begin{example}[Oscillator]
Figure~\ref{fig:osc} shows an example compound module: a simple resettable
  digital oscillator. It is composed of an inverter, an \textsc{And} gate, and a fixed unit
  delay channel, i.e., a FIFO channel with $\delta(t):=t+1$.
Both gates operate in zero-time, i.e., all delays have been lumped into the FIFO channel. 
\end{example}

\begin{figure}[h]
\centering
\begin{tikzpicture}[>=latex,semithick,scale=1.0,transform shape]

\tikzstyle{arrow}=[draw, -latex]

% boundary
\draw [rounded corners,color=blue,dashed] (1,-0.1) rectangle (5.8,1.9);

% submodules
\node[thick,rectangle,draw=black,minimum width=1cm,minimum height=0.5cm, rounded corners] at (2.2,1) (del1) {$d$};
\node[not gate US, draw] at (3.5,1) (inv1) {};
\node[and gate US, draw, scale=1.4,transform shape] at ($ (inv1)+(1.3,0.2) $) (and1) {};

% output
\node at ($ (and1)+(1.5,0) $) (o) {$\Ysig$};

\draw[->] (del1.east) -- (inv1.input);
\draw[->] (inv1.output) -- (inv1.output -| and1.input 1);
\draw[->] (and1.output) -- (o);

% input
\node at ($ (and1.input 1) + (-4,0.1) $) (i) {$\En$};
\draw[->] (i) -- ++(4,0);

% port names

%\node at ($ (3,1) - (0,0.3) $) {$x$};
\fill[black] ($ (and1.output) + (0.3,0) $) circle (2pt);
\draw[->] ($ (and1.output) + (0.3,0) $) -- ++(0,-1.0) -- ++(-4.4,0) |- (del1.west);

% submodule names

\node at ($ (inv1) - (-0.1,0.5) $) {$\textsc{Inv}$};
\node at ($ (del1) - (0,0.5) $) {$\textsc{Chn}$};
\node at ($ (and1) - (0,0.7) $) {$\textsc{And}$};

\end{tikzpicture} 
\caption{Compound oscillator module.}\label{fig:osc} 
\end{figure}
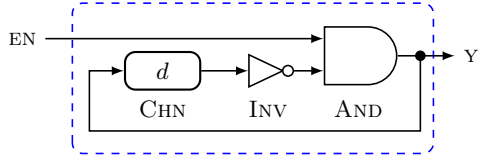

The module's behavior is characterized by the fact that it oscillates at its
  output~$\Ysig = \textsc{Chn in}$ while input~$\En$ is~1, and outputs a
  constant~$0$ while $\En$ is~0.
Up to time~$4$, the signal trace depicted in Figure~\ref{fig:exec} shows part of a
  correct execution of the oscillator module.

\definecolor{darkgreen}{rgb}{0,0.8,0}

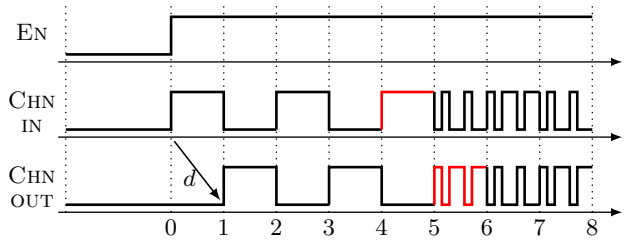
\begin{figure}[ht]
\centering
\begin{tikzpicture}[>=latex,semithick,scale=1.0,transform shape]

\tikzstyle{arrow}=[draw, -latex]

% axes
\draw[->,line width=0.5pt] (0,2) -- ++(7.5,0);
\draw[->,line width=0.5pt] (0,1) -- ++(7.5,0);
\draw[->,line width=0.5pt] (0,0) -- ++(7.5,0);

\node at ($ (-0.35,0.4)+(0,2) $) {\small \textsc{En}};

\node at ($ (-0.35,0.5)+(0,1) $) {\small \textsc{Chn}};
\node at ($ (-0.35,0.15)+(0,1) $) {\small \textsc{in}};

\node at (-0.35,0.5) {\small \textsc{Chn}};
\node at (-0.35,0.15) {\small \textsc{out}};

% times
\draw[-,line width=0.5pt,dotted] (0.1,0) -- ++(0,2.8);
\draw[-,line width=0.5pt,dotted] ($ (0.1,0)+2*(0.7,0) $) -- ++(0,2.8);
\draw[-,line width=0.5pt,dotted] ($ (0.1,0)+3*(0.7,0) $) -- ++(0,2.8);
\draw[-,line width=0.5pt,dotted] ($ (0.1,0)+4*(0.7,0) $) -- ++(0,2.8);
\draw[-,line width=0.5pt,dotted] ($ (0.1,0)+5*(0.7,0) $) -- ++(0,2.8);
\draw[-,line width=0.5pt,dotted] ($ (0.1,0)+6*(0.7,0) $) -- ++(0,2.8);
\draw[-,line width=0.5pt,dotted] ($ (0.1,0)+7*(0.7,0) $) -- ++(0,2.8);
\draw[-,line width=0.5pt,dotted] ($ (0.1,0)+8*(0.7,0) $) -- ++(0,2.8);
\draw[-,line width=0.5pt,dotted] ($ (0.1,0)+9*(0.7,0) $) -- ++(0,2.8);
\draw[-,line width=0.5pt,dotted] ($ (0.1,0)+10*(0.7,0) $) -- ++(0,2.8);

\node at ($ (0.1,0)+ 2*(0.7,0)-(0,0.2) $) {0};
\node at ($ (0.1,0)+ 3*(0.7,0)-(0,0.2) $) {1};
\node at ($ (0.1,0)+ 4*(0.7,0)-(0,0.2) $) {2};
\node at ($ (0.1,0)+ 5*(0.7,0)-(0,0.2) $) {3};
\node at ($ (0.1,0)+ 6*(0.7,0)-(0,0.2) $) {4};
\node at ($ (0.1,0)+ 7*(0.7,0)-(0,0.2) $) {5};
\node at ($ (0.1,0)+ 8*(0.7,0)-(0,0.2) $) {6};
\node at ($ (0.1,0)+ 9*(0.7,0)-(0,0.2) $) {7};
\node at ($ (0.1,0)+10*(0.7,0)-(0,0.2) $) {8};

% en signal
\draw[-,line width=1pt] ($ (0.1,2.1) + (0,0) $) --
   ++(1.4,0) -- ++(0,0.5) -- ++(5.6,0);

% signal chn in
\draw[-,line width=1pt] ($ (0.1,1.1) + (0,0) $) -- ++(1.4,0) --
   ++(0,0.5) -- ++(0.7,0) -- ++(0,-0.5) -- ++(0.7,0);
\draw[-,line width=1pt] ($ (1.5,1.1) + 2*(0.7,0) $) --
   ++(0,0.5) -- ++(0.7,0) -- ++(0,-0.5) -- ++(0.7,0);
\draw[-,line width=1pt,red] ($ (1.5,1.1) + 4*(0.7,0) $) --
   ++(0,0.5) -- ++(0.7,0);
\draw[-,line width=1pt] ($ (1.5,1.1) + 5*(0.7,0) + (0,0.5)$) --
   ++(0,-0.5) -- ++(0.1,0) --
   ++(0,0.5) -- ++(0.1,0) --
   ++(0,-0.5) -- ++(0.2,0) --
   ++(0,0.5) -- ++(0.1,0) --
   ++(0,-0.5) -- ++(0.2,0);
\draw[-,line width=1pt] ($ (1.5,1.1) + 6*(0.7,0) $) --
   ++(0,0.5) -- ++(0.1,0) --
   ++(0,-0.5) -- ++(0.1,0) --
   ++(0,0.5) -- ++(0.2,0) --
   ++(0,-0.5) -- ++(0.1,0) --
   ++(0,0.5) -- ++(0.2,0);
\draw[-,line width=1pt] ($ (1.5,1.1) + 7*(0.7,0) + (0,0.5)$) --
   ++(0,-0.5) -- ++(0.1,0) --
   ++(0,0.5) -- ++(0.1,0) --
   ++(0,-0.5) -- ++(0.2,0) --
   ++(0,0.5) -- ++(0.1,0) --
   ++(0,-0.5) -- ++(0.2,0);

% signal chan out
\draw[-,line width=1pt] ($ (0.1,0.1) + (0,0) $) -- ++(2.1,0) --
   ++(0,0.5) -- ++(0.7,0) -- ++(0,-0.5) -- ++(0.7,0);
\draw[-,line width=1pt] ($ (1.5,0.1) + 3*(0.7,0) $) --
   ++(0,0.5) -- ++(0.7,0) -- ++(0,-0.5) -- ++(0.7,0);
\draw[-,line width=1pt,red] ($ (1.5,0.1) + 5*(0.7,0) $) --
   ++(0,0.5) -- ++(0.1,0) --
   ++(0,-0.5) -- ++(0.1,0) --
   ++(0,0.5) -- ++(0.2,0) --
   ++(0,-0.5) -- ++(0.1,0) --
   ++(0,0.5) -- ++(0.2,0);
\draw[-,line width=1pt] ($ (1.5,0.1) + 6*(0.7,0) + (0,0.5)$) --
   ++(0,-0.5) -- ++(0.1,0) --
   ++(0,0.5) -- ++(0.1,0) --
   ++(0,-0.5) -- ++(0.2,0) --
   ++(0,0.5) -- ++(0.1,0) --
   ++(0,-0.5) -- ++(0.2,0);
\draw[-,line width=1pt] ($ (1.5,0.1) + 7*(0.7,0) $) --
   ++(0,0.5) -- ++(0.1,0) --
   ++(0,-0.5) -- ++(0.1,0) --
   ++(0,0.5) -- ++(0.2,0) --
   ++(0,-0.5) -- ++(0.1,0) --
   ++(0,0.5) -- ++(0.2,0);

\draw[->,shorten >=2pt, shorten <=2pt] ($ (0.1,1.0) + 2*(0.7,0) $) --
   ($ (0.1,0.1) + 3*(0.7,0) $) node[pos=0.35,below] {$d$};

\end{tikzpicture} 
\caption{Execution of oscillator module with transient channel fault
(signal mismatch marked red).}
\label{fig:exec} 
\end{figure}% is here to be in the right column

The same conceptual design of a negative feedback-loop with delay
  was used by Stricker et al.~\cite{SCBMTH08} in the context of genetic circuits, for
  synthesizing a genetic oscillator in \emph{Escherichia coli} with an output period in the
  order of an hour.
Figure~\ref{fig:osc_gen} depicts the genetic design of the feedback-loop of the simplified (second) design
   proposed by Stricker et al.
It shows the DNA segement that is introduced into the bacterial host.
The DNA comprises of a promoter (bold arrow in the figure) and a downstream lacI gene
  (flanked by a ribosome binding site and a terminator that are not shown for simplicity).
The lacI gene is transcribed and translated into LacI protein.
The promoter is activated if no inhibiting LacI proteins are present (shown as
  an inhibitory arrow from lacI to the promoter) and
  externally introduced IPTG molecules are present (not shown in the figure, and assumed to be
  present throughout).
The activation of the promoter leads to transcription and subsequent translation of the downstream
  lacI gene, resulting in increasing LacI protein levels, which then inactivate the promoter.
Only when the concentration of the LacI protein has fallen to a sufficiently low level due to degradation and dilution,
  the promoter becomes active again.
The result is an oscillation of the LacI protein concentration.

\tikzset{middlearrow/.style={
        decoration={markings,
            mark= at position 0.5 with {\arrow{#1}} ,
        },
        postaction={decorate}
    }
}

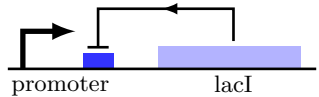
\begin{figure}[b]
\centering
\begin{tikzpicture}[>=latex,semithick,scale=1.0,transform shape]

\tikzstyle{arrow}=[draw, -latex]

% promoter
\draw[->,line width=2pt] (0,0) -- ++(0,0.5) -- ++(0.7,0);
\draw[draw=none,fill=blue!80,line width=1pt] (0.8,0) rectangle (1.2,0.2); 

% gene
\draw[draw=none,fill=blue!30,line width=1pt] (1.8,0) rectangle ($ (2.2,0.3) + (1.5,0) $); 

% base line
\draw[line width=1pt] (-0.2,0) -- (3.9,0);

% flow control
\draw[->,>=|,shorten >=2pt,shorten <=2pt,line width=1pt,middlearrow={latex}] (2.8,0.3) -- ++(0,0.5) -- ++(-1.8,0)
                                                   -- ++(0,-0.6);

\node[xshift=-0.5cm] at (1,-0.23) {\small promoter};
\node at (2.8,-0.2) {\small lacI};

\end{tikzpicture} 
\caption{Negative feedback loop of a genetic oscillator presented by Stricker et al.\ \protect\cite{SCBMTH08}.}
\label{fig:osc_gen} 
\end{figure} 

\noindent\textit{A note on state-oblivious specifications: the absence of
  an initial state.}
In the previous section, we have argued that, when specifying the channel input-output
  behavior in a state-oblivious way, we resort to input output signals as
  functions $\R \to [0,1]$.
In this section, we followed this approach in Definitions~\ref{def:comp_module}
  and \ref{def:module}.
We will later see (in Section~\ref{sec:transientfaults}) that such specifications
are also well-suited for specifying so-called self-stabilizing
  systems.

However, state-oblivious specifications also introduce difficulties that
lead to open research questions. More specifically, a useful vehicle 
for showing that some module implements another one in classical
  state-based frameworks is by induction on a sequence of input events, starting 
from some initial state.
Simulation and bi-simulation relations are proved this way, with implications
  on what can be said about using one module instead of the other.
These proof techniques fail in our case, however, since signals are defined
  on the time domain $\R$, without an initial time and ``state''.
While one can argue that induction from a common time, say~$0$, can be done
  into the positive an negative direction, questions about what this means
  for one module replacing/implementing another are open.

\section{Modeling Permanent Faults}
\label{sec:permanentfaults}

Viewed from an outside perspective, all that a faulty module can do is deviating from its
specification. Depending on the type of faults considered, there may or
may not be constraints on this deviation. In other words, a faulty module~$M$
simply follows a \emph{weaker} module specification than $\phi_M$, i.e., some
module specification $\bar{\phi}_M$ such that 
%\begin{equation*}
$\forall (\In_p)_{p\in I(M)}: \bar{\phi}_M\left((\In_p)_{p\in
I(M)}\right)\supseteq \phi_M\left((\In_p)_{p\in I(M)}\right)$.
%\end{equation*}

\begin{definition}[Crash and Byzantine fault types]
For the fault type \emph{crash faults}~\cite{Fischer83},
a faulty component simply ceases to operate at some
point in time. In this case, $\bar{\phi}_M((\In_p)_{p\in I(M)})$ can be
constructed from $\phi_M((\In_p)_{p\in I(M)})$ by adding, for each
$(\Out_p)_{p\in O(M)}\in \phi_M((\In_p)_{p\in I(M)})$ and each $t\in \R$, the
output signal
\begin{equation*}
\left(t' \in \R \mapsto \begin{cases}
\Out_p(t') & \mbox{if } t'< t\\
\Out_p(t) & \mbox{else }
\end{cases}
\right)_{p\in O(M)}
\end{equation*}
to $\bar{\phi}_M((\In_p)_{p\in I(M)})$.
This just keeps (``stuck-at'') the last output value before the crash.

The fault-type of \emph{Byzantine faults}~\cite{PSL80} is
even simpler to describe: The behavior of a faulty module is
  arbitrary, i.e., $\bar{\phi}_M$ is the constant function returning the set of
  \emph{all} possible output signals, irrespective of the input signal.
\end{definition}

Having defined a faulty type $\bar{\phi}_S$ for a module $S$ accordingly,
it seems obvious how to define a fault-tolerant compound
module: A compound module $M$ with submodules $S_M$ tolerates failures
of a subset
$F\subset S_M$, iff $\phi_M=\bar{\phi}_{M,F}$, where $\bar{\phi}_{M,F}$ is the
specification of the compound module in which we replace each submodule $S\in F$
by the one with specification $\bar{\phi}_S$. 

\begin{example}[Fault-tolerant 1-bit adder module]
Figure~\ref{fig:add} shows an example of a 1-bit adder module.
It is built from three (zero-time) 1-bit adder
  submodules, a (zero-time) majority voter, and FIFO channels with maximal
  delay~$d$ connecting the module's inputs to the adder submodules.
The channels account for the module's propagation delay and potentially desynchronized
  arrivals of input transitions at the submodules.
If the module inputs have been stable for~$d$ time, however, its output yields the sum of
  the two inputs, tolerating failure of any one of its three adder submodules and the associated input
  channels.
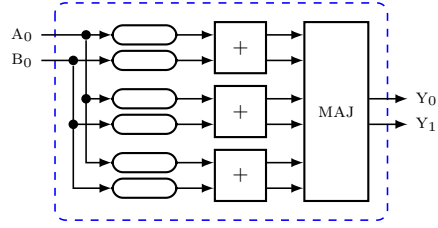
\begin{figure}[ht]
\centering
\begin{tikzpicture}[>=latex,semithick,scale=0.85,transform shape]

\tikzstyle{arrow}=[draw, -latex]

% boundary
\draw [rounded corners,color=blue,dashed] (-1.4,0.5) rectangle (3.8,-2.9);

% submodules
\node[thick,rectangle,draw=black,minimum width=1cm,minimum height=0.3cm, rounded corners] at (0, 0) (del1) {};
\node[thick,rectangle,draw=black,minimum width=1cm,minimum height=0.3cm, rounded corners] at (0,-0.4) (del2) {};
\node[thick,rectangle,draw=black,minimum width=1cm,minimum height=0.3cm, rounded corners] at (0,-1) (del3) {};
\node[thick,rectangle,draw=black,minimum width=1cm,minimum height=0.3cm, rounded corners] at (0,-1.4) (del4) {};
\node[thick,rectangle,draw=black,minimum width=1cm,minimum height=0.3cm, rounded corners] at (0,-2) (del5) {};
\node[thick,rectangle,draw=black,minimum width=1cm,minimum height=0.3cm, rounded corners] at (0,-2.4) (del6) {};

\node[thick, rectangle,draw=black,minimum width=0.8cm,minimum height=0.8cm] at ($ (del1) + (1.5,-0.2)$) (add1) {+};
\node[thick, rectangle,draw=black,minimum width=0.8cm,minimum height=0.8cm] at ($ (del3) + (1.5,-0.2)$) (add2) {+};
\node[thick, rectangle,draw=black,minimum width=0.8cm,minimum height=0.8cm] at ($ (del5) + (1.5,-0.2)$) (add3) {+};

\node[thick, rectangle,draw=black,minimum width=1cm,minimum height=2.8cm] at ($ (add2) + (1.5,0)$) (vote1) {\textsc{maj}};

%\node[and gate US, draw, scale=1.4,transform shape] at ($ (inv1)+(1.3,0.2) $) (and1) {};

% output
%\node at ($ (and1)+(1.5,0) $) (o) {$\Ysig$};

\draw[->] (del1.east) -- (del1.east -| add1.west);
\draw[->] (del2.east) -- (del2.east -| add1.west);
\draw[->] (del3.east) -- (del3.east -| add2.west);
\draw[->] (del4.east) -- (del4.east -| add2.west);
\draw[->] (del5.east) -- (del5.east -| add2.west);
\draw[->] (del6.east) -- (del6.east -| add2.west);

\draw[->] (del1.east -| add1.east) -- (del1.east -| vote1.west);
\draw[->] (del2.east -| add1.east) -- (del2.east -| vote1.west);
\draw[->] (del3.east -| add2.east) -- (del3.east -| vote1.west);
\draw[->] (del4.east -| add2.east) -- (del4.east -| vote1.west);
\draw[->] (del5.east -| add3.east) -- (del5.east -| vote1.west);
\draw[->] (del6.east -| add3.east) -- (del6.east -| vote1.west);

% output
\draw[->] (del3.east -| vote1.east) -- ++(0.6,0) node (y0) {};
\draw[->] (del4.east -| vote1.east) -- ++(0.6,0) node (y1) {};

\node at ($ (y0) + (0.3,0) $) {$\Ysig_0$};
\node at ($ (y1) + (0.3,0) $) {$\Ysig_1$};

% input
\draw[<-] (del1.west) -- ++(-1.1,0) node (x0) {};
\draw[<-] (del2.west) -- ++(-1.1,0) node (x1) {};

\node at ($ (x0) + (-0.3,0) $) {$\Asig_0$};
\node at ($ (x1) + (-0.3,0) $) {$\Bsig_0$};

\node[fill,black,inner sep=1.5pt,circle] at ($ (del1.west) + (-0.4,0) $) (circ1) {};
\node[fill,black,inner sep=1.5pt,circle] at ($ (del2.west) + (-0.6,0) $) (circ2) {};
\draw[->] (circ1) |- (del3.west);
\draw[->] (circ2) |- (del4.west);

\node[fill,black,inner sep=1.5pt,circle] at (circ1 |- del3.west) (circ3) {};
\node[fill,black,inner sep=1.5pt,circle] at (circ2 |- del4.west) (circ4) {};
\draw[->] (circ1 |- del3.west) |- (del5.west);
\draw[->] (circ2 |- del4.west) |- (del6.west);

%\node at ($ (and1.input 1) + (-4,0.1) $) (i) {$\En$};
%\draw[->] (i) -- ++(4,0);

% port names

%\node at ($ (3,1) - (0,0.3) $) {$x$};
%\fill[black] ($ (and1.output) + (0.3,0) $) circle (2pt);
%\draw[->] ($ (and1.output) + (0.3,0) $) -- ++(0,-1.0) -- ++(-4.4,0) |- (del1.west);

% submodule names

%\node at ($ (inv1) - (-0.1,0.5) $) {$\textsc{Inv}$};
%\node at ($ (del1) - (0,0.5) $) {$\textsc{Chn}$};
%\node at ($ (and1) - (0,0.7) $) {$\textsc{And}$};

\end{tikzpicture} 
\caption{Fault-tolerant 1-bit adder.}\label{fig:add} 
\end{figure}
\end{example}

While this definition of a fault-tolerant compound module can be useful, it is very restrictive.
For instance, our adder compound module cannot tolerate a failure of the majority
  voter that computes the output.
More generally, no matter how a compound module~$M$ is constructed, it can never tolerate
  even a single crash failure of
  an arbitrary submodule $S$, unless $M$ is trivial:
  If $S$ has an output port in common with $M$, i.e., if $S$ generates
  this output for $M$, the only possible guarantee $M$
  could make for this output port is a fixed output value at all times,
  as this is what the crash of $S$ would lead to. 

To address this issue, we introduce the concept of a
  \emph{fault-tolerant implementation} of a module.

\begin{definition}[Fault-tolerant implementation]
We say that module $M$ \emph{implements} module $M'$ iff
\begin{equation*}
  \forall (\In_p)_{p\in I(M)}: \phi_M\left((\In_p)_{p\in
  I(M)}\right)\subseteq \phi_{M'}\left((\In_p)_{p\in I(M')}\right)\,.
\end{equation*}
This requires that $I(M)=I(M')$ and $O(M)=O(M')$.
Similarly, for a given fault type $\bar{\cdot}$, $M$ is an \emph{implementation
  of $M'$ that tolerates failures of $F\subset S_M$} iff
\begin{equation*}
  \forall (\In_p)_{p\in I(M)}: \bar{\phi}_{M,F}\left((\In_p)_{p\in
  I(M)}\right)\subseteq \phi_{M'}\left((\In_p)_{p\in I(M')}\right)\,,
\end{equation*}
  where $\bar{\phi}_{M,F}$ is defined according to the fault type.
Finally, $M$ is an \emph{$f$-tolerant implementation of $M'$}, iff it tolerates
  faults of $F\subset S_M$ for any $F$ satisfying $|F|\leq f$.
\end{definition}

\begin{example}[Fault-tolerant adder]
For an adder implementation that is $1$-toler\-ant to Byzantine faults (and thus
also any other fault type), \emph{triple-modular redundancy} (TMR) can be used.
Here, not just the adders, but also the pair of input and output signals is
triplicated. Moreover, the single majority voter at the adder outputs is
replaced by three majority voters at the adder inputs: Since they vote on
the replicated input signals, we can guarantee that all three adders receive
identical inputs if no voter fails, whereas two adders receive identical inputs
and produce identical outputs if one voter is faulty. Note that relaxing the
specification and using an implementation relation is necessary here, as
otherwise the same reasoning as before would prevent a $1$-tolerant solution.
\end{example}

\pgfdeclareimage[width=8.5cm]{dartsnode}{DARTS_ASIC_node}

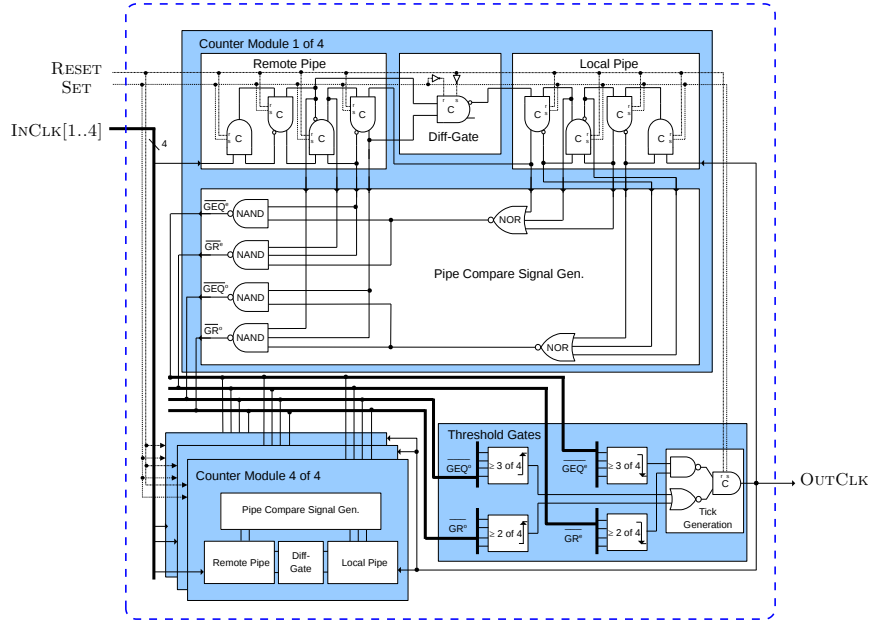
\begin{figure*}[ht]
\centering
\begin{tikzpicture}[>=latex,semithick,scale=1.31,transform shape]

\tikzstyle{arrow}=[draw, -latex]

% darts node
\node at (0,0) (darts) {\pgfuseimage{dartsnode}};

% boundary
\draw [rounded corners,color=blue,dashed] 
  ($ (darts.south west) + (0.7,0.1) $) rectangle ($ (darts.north east) + (-1.45,0.1) $);

% in/output
\node at ($ (darts.south west) + (0.7,0.1) + (-0.5,5.6) $) {\tiny \textsc{Reset}};
\node at ($ (darts.south west) + (0.7,0.1) + (-0.5,5.4) $) {\tiny \textsc{Set}}; 
\node at ($ (darts.south west) + (0.7,0.1) + (-0.7,4.95) $) {\tiny \textsc{InClk[1..4]}};
\node at ($  (darts.north east) + (-1.45,0.1) + (0.6,-4.83) $) {\tiny \textsc{OutClk}};

\end{tikzpicture} 
\caption{Node of the fault-tolerant DARTS oscillator.}\label{fig:darts} 
\end{figure*}

The problem of developing a fault-tolerant implementation of an oscillator was
addressed in the DARTS project~\cite{FS11:JECE,FS12}: Using a predecessor of the
proposed modeling framework, it was shown that a circuit comprising~$3f+1$ tick generator
nodes, in which the output of each node is fed back as input to all nodes
(including the node itself), tolerates up to~$f$ Byzantine faulty
nodes. The circuitry of a single tick generator node for $n=4$ and $f=1$
is depicted in Figure~\ref{fig:darts}.
Informally, it counts the difference of clock transitions generated by itself (Local
Pipe) and those received from other nodes (Remote Pipe) by means of Counter
Modules implemented via elastic pipelines \cite{Sutherland89}. 
When sufficiently many nodes (not all, since there might be a fault) are not too
far behind (determined by the Threshold Gates), it generates a new local clock
transition (Tick Generation).

\section{Modeling Transient Faults and Self-Stabilization}
\label{sec:transientfaults}

\emph{Transient faults} are assumed to be temporary in nature, in the sense that
  the \emph{cause} of the fault eventually vanishes.
Suitable recovery techniques can hence be used to resume correct operation later on.
The most extreme requirement is \emph{self-stabilization}~\cite{dijkstra74}, where the
  system must eventually resume correct operation
  even after all of its components experienced arbitrary transient faults.
Since the latter results in arbitrary states, this is equivalent to
  requiring that the system re-establishes correct operation from arbitrary initial
  states in finite time.
The maximal time it may take to do so is called \emph{stabilization time}.

Self-stabilization plays a crucial role in mission-critical systems, where
  even the assumption that a certain fraction (e.g., less than a third, as in
  DARTS) of the subcomponents can fail is too optimistic, 
  or for applications that cannot afford the amount of redundancy needed
  for fully masking faults.
Unsurprisingly, self-stabilization and related concepts also play a vital role in biological
  systems.
For example, Albert and Othmer~\cite{AO03} modeled part of the
  control circuit that regulates gene expression in the fruit fly
  Drosophila Melanogaster by a binary circuit, and observed that a
  considerable number of initial circuit states finally lead to the
  wild-type stable state.
For the lobster heart, it has been established that it is,
  in essence, self-stabilizing even in the presence of ongoing Byzantine
  behavior~\cite{DDP03:SSS,sivan99lobster}.

Transferring the idea of self-stabilization to our state-oblivious framework
  requires some effort, but we will demonstrate that it can be integrated
  very well.
Most notably, there is no notion of a state (besides from trivial states),
which means that we cannot define self-stabilization in the conventional
  state-based manner.

The first important observation underlying input-output functions in a state-oblivious
  specification of self-stabilization is that, since a basic module specification
  describes the desired behavior \emph{from the viewpoint of an external observer},
  we consider a module correct even if it merely \emph{seems} to be operating correctly in
  terms of its input-output behavior.
In other words, it does not matter whether the module internally operates as intended,
  as long as it produces correct results.

Second, when defining basic modules (Definition~\ref{def:module}), we
  only resorted to input and output signals from $\R \to [0,1]$.
The input-output function $\phi_M$ of a module~$M$ then maps input
  signals to allowed output signals.
For modules that are intended to be self-stabilizing (or that suffer from transient faults),
  this is not anymore convenient since they are not either correct or faulty during all
  of the execution.
Merely, we would like to define how they should behave if they
  were \emph{correct during a time interval $[t^-,t^+]$}.

\noindent\textit{Redefining correctness for transient faults.}
An immediate solution to this is to define $\phi_M$ on all
signal restrictions to all sub-intervals $I=[t^-,t^+] \subseteq \R$.
To make such interval-restrictions explicit, we will sometimes write
  $\sigma^I$, $\In^I$, $\Out^I$, $\E^I$ etc.
Note that such intervals $I$ could also be open $(t^-,t^+)$, closed
  $[t^-,t^+]$, or half-open, but are always contiguous. 

  \begin{definition}[Basic module---interval-restricted specification]
An in\-terval-restricted execution $E^I$ of a basic module is correct during
  $I=[t^-,t^+]$ if its interval-restricted output signals $\Out^I$
  are within the image $\phi_M(\In^I)$ of its interval-restricted input
  signals.
\end{definition}
We termed this the \emph{interval-restricted specification}, since it requires the
  definition of $\phi_M$ on all these sub-intervals.

However, care has to be taken: the input-output function $\phi_M$ has to be
  restricted to ensure that it adheres
  to an intuitive notion of correctness.
For example, we expect an execution of $M$ that is correct
  within $[t^-,t^+]$ to be also correct within all subintervals
  of $[t^-,t^+]$.
One possibility is to add all those restrictions explicitly.
Indeed, such specifications are powerful in
  expressiveness~\cite{dolev14}, but at the same time 
  our experience in the early stages of~\cite{dolev14} was that using this approach is
  tedious for simple modules, and practically guarantees mistakes for complex modules.
The reason for this is that the subset-closedness of the correctness definition
  is easily violated in an interval-restricted specification even for simple modules like
  channels.

We thus primarily resort to another definition, which does not change the
  domain of~$\phi_M$, i.e, where the domain of all signals is $\R$,
  which we call a \emph{definition by extension}.
With this definition, correctness is subset-closed for time intervals by construction
  (see \cite[Lem.~3.3]{dolev14}), such that the natural notion of correctness
  is also guaranteed by construction.

\begin{definition}[Basic module, correct during time interval---definition by extension]
An interval-restricted execution $E^I$ of a basic module $M$ is correct during $I=[t^-,t^+]$, iff
  there is a (complete, i.e., with time domain $\R$)
  execution $E'$ of a basic module $M$ such that:
(i) the input and output signals of $E^I$ and $E'$ are identical during $[t^-,t^+]$, and
(ii) for execution~$E'$, letting $i$ be the input signals and $o$ the output signals of $E'$,
  $o \in \phi_M(i)$.
\end{definition}

If not stated otherwise, we will resort to the definition by extension for basic modules.
For ease of notation, we will, however, extend $\phi_M$ to input and output
  signals with time domains that are sub-intervals $[t^-,t^+]$ of $\R$ in the following:
Writing $\Out^I \in \phi_M(\In^I)$ where $\In^I,\Out^I$ have time domain $I=[t^-,t^+]$
  is just a short-hand notation for: For any execution $E$ of $M$ that behaves according to $\In^I,\Out^I$
  during $I$, basic module $M$ is correct during $I$.

\begin{definition}[Extendible module]
\label{def:extendible}
We say that a basic module is \emph{extendible}, if its input-output function $\phi_M$
  has the properties of a definition by extension.
That is, executions that are correct on a subinterval can be extended
  to executions that are correct on $\R$.
\end{definition}

The above definition of correctness introduced above also implies that a sub-execution
  on some interval $[t^-,t^+]\subset \R$ that is considered correct can be extended to a
  (complete) correct execution on $\R$.
This is natural for basic modules, but inappropriate for
  self-stabilizing compound modules: these take the role of algorithms,
  and making this a requirement would be equivalent to disallowing transient
  faults---or, more precisely, to implicitly turn them into persistent faults.
To illustrate this issue, consider again the compound module implementing the
  oscillator shown in Figure~\ref{fig:osc} and the signal trace shown in
  Figure~\ref{fig:exec}:
The execution segment during time interval~$[6,8]$ must
  be considered correct, since it fulfills all input-output constraints of the
  involved circuit components during this interval.
We know, however, that such a high-frequency oscillation can never occur in a
  (complete) correct execution of the compound module on $\R$,
  for which the only possible oscillator frequency is one transition per time
  unit.

To allow for a meaningful notion of self-stabilization, we will hence treat
  compound modules differently: When analyzing their stabilizing behavior,
  we assume that all sub-modules themselves operate correctly, whereas the 
  ``convergence'' of the compound module's externally visible behavior to a correct
  one must be enforced.
This is captured by defining correct executions of compound
  modules on time intervals $[t^-,t^+]\subset \R$ by the same process as in
  Definition~\ref{def:comp_module}, except that we replace $\R$ by $[t^-,t^+]$.

\begin{definition}[Compound module---interval-restriced specification]\ \\
For any $I=[t^-,t^+]$ and any interval-restricted input signal
$(\In_p^I)_{p\in I(M)}$, we require that the interval-restricted output signal
$(\Out_p^I)_{p\in O(M)}\in \phi_M((\In_p^I)_{p\in I(M)})$ iff
there exist interval-restricted input and output signals $(f_p^I)_{p\in \bigcup_{S\in S_M} I(S)\cup O(S)}$
for all submodules $S$ of $M$ so that all the properties below hold:
\begin{compactitem}
\item $\forall S\in S_M: (f_p^I)_{p\in O(S)} \in \phi_S((f_p^I)_{p\in I(S)})$
\item $\forall p\in I(M): f_p=\In_p^I$
\item $\forall p\in O(M): f_p=\Out_p^I$
\end{compactitem}
\end{definition}

Note that this definition is recursive; we can iteratively extend all module
  specifications to inputs on arbitrary intervals $[t^-,t^+]\subseteq \R$,
  starting from the specifications of basic modules for inputs on $\R$.

With these definitions in place, we can now proceed to defining a suitable
  notion of self-stabilization in our framework.

\begin{definition}[Self-stabilizing implementation]\label{def:selfstab}
A module $M$ is called a
\emph{$T$-stabilizing implementation of module $M'$}, iff $I(M)=I(M')$,
$O(M)=O(M')$ and, for all
$I=[t^-,t^+]\subseteq \R$ with $t^+\geq t_-+T$, $I'=[t^-+T,t^+]$ and each
$(\Out_p^I)_{p\in O(M)}\in  \phi_M\left( (\In_p^I)_{p\in I(M)} \right)$,
it holds that 
\[
(\Out_p^{I'})_{p\in O(M)} \in  \phi_{M'}\left( (\In_p^{I'})_{p\in I(M')} \right).
\]
Informally, cutting off the first $T$ time units from any interval-restricted
execution of $M$ must yield a correct interval-restricted execution of $M'$.

Module $M$ is a \emph{self-stabilizing implementation of $M'$}, iff
  it is a $T$-stabilizing implementation of $M'$ for some $T<\infty$.
\end{definition}

\begin{example}[Self-stabilization]
According to \cref{def:selfstab}, the oscillator implementation from
  Figure~\ref{fig:osc} is not self-stabilizing,
  as illustrated by Figure~\ref{fig:exec}:
After a transient fault of the channel component during time~$[5,6]$, all circuit components
  operate correctly again from time~$6$ on, but the behavior of circuit
  output $\Ysig = \textsc{Chn in}$ never returns to the behavior
  of~$\Ysig$ that could be observed in an execution on~$\R$.

For a positive example, recall the 1-bit adder depicted in Figure~\ref{fig:add}.
Its self-stabilization properties follow, without the need of a custom analysis,
  from a general principle (called forgetfulness), which will be introduced in the
  next section.
\end{example}

\section{Example: A Self-Stabilizing Oscillator}
\label{sec:SSOsc}

In view of the state-obliviousness and generality of our modeling framework,
  one might ask whether it indeed allows to derive meaningful
  results.
As a proof of concept, we will thus elaborate more on 
  self-stabilizing compound modules.

First, we will formalize the statement that if a compound module~$M$ is
  made up of submodules~$S$ whose output at time $t$ depends only on the input
  during $[t-T_S,t]$ (for $T_S\in \R_0^+$) and contains no feedback-loops
  (like, e.g., the adder in Figure~\ref{fig:add}, but unlike the oscillator in
  Figure~\ref{fig:osc}), then $M$ is self-stabilizing.
Interestingly, 
  this result sometimes \emph{does} also apply to systems that do have internal
  feedback loops; this
  holds true whenever we can contain the loop in a submodule and (separately) show
  that it is self-stabilizing.
\begin{definition}[Forgetfulness]\label{def:forgetfulness}
For $F \geq 0$, module $M$ is \emph{$F$-forgetful} iff:
\begin{compactenum}
  \item For any $I=[t^-,t^+] \subseteq \R$ with $t^+ \geq t^-+F$, pick any
  interval-restricted output $( \Out_p^I)_{p\in O(M)} \in  \phi_M\left( (\In_p^I)_{p\in I(M)} \right)$.

  \item For each input port $p\in I(M)$, pick any input signal
    $\In_p':\R\to \{0,1\}$ so that $\In_p'$
    restricted to $I$ equals $\In_p^I$.

  \item Then
  %\begin{align*}
  $\left( \Out_p':\R\to \{0,1\} \right)_{p\in O(M)} \in
  \phi_M\left( (\In_p':\R\to \{0,1\})_{p\in I(M)}) \right)$ %\,
  %\end{align*}
  exists so that for all output ports $p\in O(M)$ the restrictions of
    $\Out_p$ and $\Out_p'$ to the interval $[t^-+F,t^+]$ are equal.
\end{compactenum}
\end{definition}
In other words, the output of a $F$-forgetful module during $[t^-+F,t^+]$ reveals no
  information regarding the input during $(-\infty,t^-)$.

\begin{example}  
A simple example of a $d$-forgetful module is a FIFO channel with maximum
  delay~$d$.
\end{example}

\begin{definition}[Feedback-free module]\label{def:feedbackfree}
Let the \emph{circuit graph} of a compound module $M$ be the directed graph
  whose nodes are the submodules $S_M$ of $M$,
  and for each output port $p$ of $S\in S_M$ that is an input port of
  another submodule $S'\in S_M$, there is a directed edge
  from $S$ to $S'$.
We say~$M$ is \emph{feedback-free} iff all its submodules
  are forgetful and its circuit graph is acyclic.
\end{definition}

One can then show that feedback-free compound modules made up
  of forgetful submodules are self-stabilizing:

\begin{theorem}[\cite{dolev14}, Theorem 3.7]\label{theorem:forgetful}
Given a feedback-free compound module $M$, denote by ${\cal P}$ the set of paths
  in its circuit graph.
Suppose that each submodule $S\in {\cal S}_M$ is $F_S$-forgetful for
  some $F_S\in \R^+_0$.
Then, $M$ is $F$-forgetful with
\begin{equation*}
F = \max_{(S_1,\ldots,S_k)\in {\cal P}}\left\{\sum_{i=1}^k F_{S_i}\right\}\,.
\end{equation*}
\end{theorem}

Using this theorem (possibly recursively applied, in
the case of compound modules made up of compound submodules),
one can show that a given
 feedback-free compound module is forgetful. Moreover, for such a module $M$,
 it is sufficient to show that it behaves like
 another module $M'$ in correct executions (i.e., those on $\R$, rather than
 on certain time intervals) for proving that $M$ is a self-stabilizing implementation of $M'$.

\begin{corollary}\label{cor:s}
Suppose that compound module $M$ satisfies the prerequisites of
  Theorem~\ref{theorem:forgetful}.
Moreover, for a module $M'$ with $I(M')=I(M)$ and $O(M')=O(M)$,
  assume that:
For all input signals $(\In_p:\R\to \{0,1\})_{p\in I(M)}$, it holds
for its output signals that
%\begin{align*}
$\phi_M\left( (\In_p)_{p\in I(M)} \right) \subseteq
\phi_{M'}\left( (\In_p)_{p\in I(M)} \right)$.
%\end{align*}
Then, $M$ is a self-stabilizing implementation of $M'$.
\end{corollary}

These results ensure that self-stabilization follows without further
  ado not only in trivial cases where an erroneous state is
  \emph{instantaneously} forgotten and overwritten by new input.
By using compound modules in a hierarchical manner, one can
  encapsulate the heart of a proof of self-stabilization in the appropriate
  system layer and separate aspects from different layers that are unrelated.
This plays along nicely with the standard approach for proving self-stabilization of
  complex systems, which is to establish properties of increasing strength and
  complexity in a bottom-up fashion, advancing from very basic aspects to high-level 
  arguments.

\noindent\textit{Beyond feedback-free compound modules.}
Unfortunately, however, it is not hard to see that if (the circuit graph of) a
  compound module $M$ is not feedback-free, self-stabilization of $M$ does not necessarily
  follow from the fact that all submodules are forgetful.
An example of such a circuit is the oscillator in Figure~\ref{fig:osc}.
There are, however, circuits with feedback loops that stabilize.

\begin{figure}[b]
\centering
\begin{tikzpicture}[>=latex,semithick,scale=1.0,transform shape]

\tikzstyle{arrow}=[draw, -latex]

% boundary
\draw [rounded corners,color=blue,dashed] (0.9,-0.2) rectangle (5.8,1.6);

% submodules
\node[thick,rectangle,draw=black,minimum width=1cm,minimum height=0.5cm, rounded corners] at (2.2,1) (del1) {$d$};
\node[not gate US, draw] at (3.5,1) (inv1) {};
\node[thick, rectangle,draw=black,minimum width=0.8cm,minimum height=0.8cm] at ($ (inv1)+(1.3,0) $) (mem) {};
\node at ($ (mem.center)+(0.15,-0.15) $) {\VarClock}; %{\showclock{0}{45}}; 

% output
\node at ($ (mem)+(1.5,0) $) (o) {$\Ysig$};

\draw[->] (del1.east) -- (inv1.input);
\draw[->] (inv1.output) -- (mem.west);
\draw[->] (mem.east) -- (o);

%\node at ($ (3,1) - (0,0.3) $) {$x$};
\fill[black] ($ (mem.east) + (0.3,0) $) circle (2pt);
\draw[->] ($ (mem.east) + (0.3,0) $) -- ++(0,-1.0) -- ++(-4.4,0) |- (del1.west);

% submodule names

\node at ($ (inv1) - (-0.1,0.5) $) {$\textsc{Inv}$};
\node at ($ (del1) - (0,0.5) $) {$\textsc{Chn}$};
\node at ($ (mem) - (0,0.7) $) {$\textsc{Mem}$};
\node at ($ (mem) - (0.7,-0.27) $) {\small \sc X};

\end{tikzpicture} 
\caption{Self-stabilizing oscillator module.}\label{fig:osc_mem} 
\end{figure}
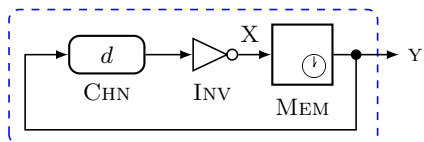

\begin{example}
Figure~\ref{fig:osc_mem} shows  a self-stabilizing variant of the
  oscillator from Figure~\ref{fig:osc} (without enable input).
The module consists of (i) a watchdog-timed memory-cell \textsc{Mem}
  whose output~$\textsc{Mem Out} = Y$;
  the output is~$1$ at time~$t$ iff there is a time
  $t'\in (t-T,t)$ where output~$Y(t')=0$ and input~$X(t')=1$,
(ii) a succeeding fixed delay channel with delay~$d \le T$, and
(iii) an inverter. Note that (i) implies that the set 
$\{t\in \R_0^+\,|\,Y(t)=0\}$ is closed.
\end{example}

We now demonstrate how to formalize this example and its proof in our state-oblivious
  framework.
The (basic) module \textsc{osc} has no input and one output $Y$.
Recall that, for basic modules, specifications involve defining $\phi_M$ on
  executions on $\R$ only.
Hence, the module specification $\phi_{osc}$ is fully defined by deciding whether
  some function
$(\Out_Y:\R\to \R) \in \phi_{osc}(\emptyset)$ or not.
We define this to be the case for all functions satisfying that
  $\exists \delta\in [0,T+d)$ so that
\begin{equation*}
\Out_Y(t) = \begin{cases}
1 & \text{if }\exists z\in \mathbb{Z}, \exists \tau \in (0,T): t=\delta+z(T+d)+\tau\\
0 & \text{else}\,.
\end{cases}
\end{equation*}
Intuitively, $\delta$ denotes the fixed time offset of the signal, and $\tau$ the
  time the signal is 1 during the period $t+d$.

If restricted to times $[0,5]$,
  Figure~\ref{fig:exec_mem} shows the execution for $\delta=0$, $d=1$ and
  $T=1.5d=1.5$.
The channel incorrectly forwards the red signal at its input
  during time~$[4,5]$ to the red signal at its output during~$[5,6]$, i.e., is
  not correct during $[5,6]$.
However, the module quickly recovers:
Starting from time $6.3$, the circuit has returned to a
  feasible periodic behavior with $\delta=0.3$ and $\tau=T$.
We next show that this stabilizing behavior is guaranteed.

\definecolor{darkgreen}{rgb}{0,0.8,0}

\newcommand{\T}{1.05}
\newcommand{\de}{0.7}

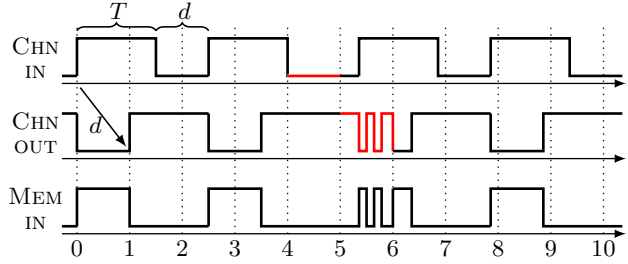
\begin{figure}[ht]
\centering
\begin{tikzpicture}[>=latex,semithick,scale=1.0,transform shape]

\tikzstyle{arrow}=[draw, -latex]

% axes
\draw[->,line width=0.5pt] (0,2) -- ++(7.5,0);
\draw[->,line width=0.5pt] (0,1) -- ++(7.5,0);
\draw[->,line width=0.5pt] (0,0) -- ++(7.5,0);

\node at ($ (-0.35,0.5)+(0,2) $) {\small \textsc{Chn}};
\node at ($ (-0.35,0.15)+(0,2) $) {\small \textsc{in}};

\node at ($ (-0.35,0.5)+(0,1) $) {\small \textsc{Chn}};
\node at ($ (-0.35,0.15)+(0,1) $) {\small \textsc{out}};

\node at ($ (-0.35,0.5)+(0,0) $) {\small \textsc{Mem}};
\node at ($ (-0.35,0.15)+(0,0) $) {\small \textsc{in}};

% times
\draw[-,line width=0.5pt,dotted] (0.2,0) -- ++(0,2.8);
\draw[-,line width=0.5pt,dotted] ($ (0.2,0)+1*(0.7,0) $) -- ++(0,2.8);
\draw[-,line width=0.5pt,dotted] ($ (0.2,0)+2*(0.7,0) $) -- ++(0,2.8);
\draw[-,line width=0.5pt,dotted] ($ (0.2,0)+3*(0.7,0) $) -- ++(0,2.8);
\draw[-,line width=0.5pt,dotted] ($ (0.2,0)+4*(0.7,0) $) -- ++(0,2.8);
\draw[-,line width=0.5pt,dotted] ($ (0.2,0)+5*(0.7,0) $) -- ++(0,2.8);
\draw[-,line width=0.5pt,dotted] ($ (0.2,0)+6*(0.7,0) $) -- ++(0,2.8);
\draw[-,line width=0.5pt,dotted] ($ (0.2,0)+7*(0.7,0) $) -- ++(0,2.8);
\draw[-,line width=0.5pt,dotted] ($ (0.2,0)+8*(0.7,0) $) -- ++(0,2.8);
\draw[-,line width=0.5pt,dotted] ($ (0.2,0)+9*(0.7,0) $) -- ++(0,2.8);
\draw[-,line width=0.5pt,dotted] ($ (0.2,0)+10*(0.7,0) $) -- ++(0,2.8);

\node at ($ (0.2,0)+ 0*(0.7,0)-(0,0.2) $) {0};
\node at ($ (0.2,0)+ 1*(0.7,0)-(0,0.2) $) {1};
\node at ($ (0.2,0)+ 2*(0.7,0)-(0,0.2) $) {2};
\node at ($ (0.2,0)+ 3*(0.7,0)-(0,0.2) $) {3};
\node at ($ (0.2,0)+ 4*(0.7,0)-(0,0.2) $) {4};
\node at ($ (0.2,0)+ 5*(0.7,0)-(0,0.2) $) {5};
\node at ($ (0.2,0)+ 6*(0.7,0)-(0,0.2) $) {6};
\node at ($ (0.2,0)+ 7*(0.7,0)-(0,0.2) $) {7};
\node at ($ (0.2,0)+ 8*(0.7,0)-(0,0.2) $) {8};
\node at ($ (0.2,0)+ 9*(0.7,0)-(0,0.2) $) {9};
\node at ($ (0.2,0)+10*(0.7,0)-(0,0.2) $) {10};

\newcommand{\newPer}[1]{
% signal chn in
\draw[-,line width=1pt]
   ($ #1 + (0,2.1) + (0,0) $) 
   -- ++(0,0.5) -- ++(\T,0)
   -- ++(0,-0.5) -- ++(\de,0);
% chn out
\draw[-,line width=1pt]
   ($ #1 + (0,1.1) + (0,0.5) $) 
   -- ++(0,-0.5) -- ++(\de,0)
   -- ++(0,0.5) -- ++(\T,0);
% mem in
\draw[-,line width=1pt]
   ($ #1 + (0,0.1) + (0,0) $) 
   -- ++(0,0.5) -- ++(\de,0)
   -- ++(0,-0.5) -- ++(\T,0);
}

% before period:
% signal chn in
\draw[-,line width=1pt]
   ($ (0,2.1) $) -- ++(0.2,0);
% chn out
\draw[-,line width=1pt]
   ($ (0,1.1) + (0,0.5) $) -- ++(0.2,0);
% mem in
\draw[-,line width=1pt]
   ($ (0,0.1) $) -- ++(0.2,0);

% period
\newPer{(0.2,0)}
\newPer{(1.95,0)}

% fault:
% signal chn in
\draw[-,line width=1.5pt,color=white] % some hack to remove old line
   ($ (1.95,0) + (0,2.1) + (\T,0) $) -- ++(\de,0);
\draw[-,line width=1pt,color=red]
   ($ (1.95,0) + (0,2.1) + (\T,0) $) -- ++(\de,0);
% chn out
\draw[-,line width=1pt,color=red]
   ($ (1.95,0) + (0,1.1) + (0,0.5) + (\T,0) + (\de,0) $) -- ++(0.25,0)
   -- ++(0,-0.5) -- ++(0.1,0)
   -- ++(0,0.5) -- ++(0.1,0)
   -- ++(0,-0.5) -- ++(0.1,0)
   -- ++(0,0.5) -- ++(0.15,0)
   -- ++(0,-0.5);
% mem in
\draw[-,line width=1pt]
   ($ (1.95,0) + (0,0.1) + (0,0) + (\T,0) + (\de,0) $) -- ++(0.25,0)
   -- ++(0,0.5) -- ++(0.1,0)
   -- ++(0,-0.5) -- ++(0.1,0)
   -- ++(0,0.5) -- ++(0.1,0)
   -- ++(0,-0.5) -- ++(0.15,0)
   -- ++(0,0.5);

% after the fault:
% signal chn in
\draw[-,line width=1pt]
   ($ (1.95,0) + (0,2.1) + (\T,0) + (\de,0)$) -- ++(0.25,0);

% shifted period starts
% signal chn in
\draw[-,line width=1pt]
   ($ (1.95,0) + (0,2.1) + (\T,0) + (\de,0) + (0.25,0) $) 
   -- ++(0,0.5) -- ++(\T,0)
   -- ++(0,-0.5) -- ++(\de,0);
% chn out
\draw[-,line width=1pt]
   ($  (1.95,0) + (0,1.1) + (\T,0) + 2*(\de,0) $) 
   -- ++(0.25,0)
   -- ++(0,0.5) -- ++(\T,0);
% mem in
\draw[-,line width=1pt]
   ($  (1.95,0) + (0,0.1) + (\T,0) + 2*(\de,0) + (0,0.5) $) 
   -- ++(0.25,0)
   -- ++(0,-0.5) -- ++(\T,0);

% normal period
\newPer{(1.95,0) + 2*(\T,0) + 2*(\de,0) + (0.25,0)}

\draw[->,shorten >=2pt, shorten <=2pt] ($ (0.2,2.0) $) --
   ($ (0.2,1.1) + (\de,0) $) node[pos=0.35,below] {$d$};

\draw [decorate,decoration={brace,amplitude=3pt},xshift=0,yshift=0]
($  (0.2,2.1) + (0,0.6) $) -- ($  (0.2,2.1) + (0,0.6) + (\T,0) $)
node[black,midway,yshift=0.27cm,xshift=0.03cm] {\footnotesize $T$};

\draw [decorate,decoration={brace,amplitude=3pt},xshift=0,yshift=0]
($  (0.2,2.1) + (0,0.6) + (\T,0) $) -- ++(\de,0)
node[black,midway,yshift=0.27cm,xshift=0.03cm] {\footnotesize $d$};

\end{tikzpicture} 
\caption{Execution part of self-stabilizing oscillator module
  with transient channel fault (incorrect propagation in red).}
\label{fig:exec_mem} 
\end{figure}

\begin{lemma}
If $T\geq d$, the compound module given in Figure~\ref{fig:osc_mem} is a
  $(T+2d)$-stabilizing implementation of \textsc{osc}.
\end{lemma}
\begin{proof}
Wlog., assume that the compound module follows its specification during
$[t^-,t^+)=[0,\infty)$. 
\begin{compactenum}
  \item There is some time $t^*\in [0,T+d]$ when $Y(t^*)=0$. (Otherwise, the
  input to \textsc{Mem} at every time $t\in [d,T+d]$ would be
  $X(t)=\overline{Y}(t-d)=0$ (with $\overline{Y}$ denoting negation),
  entailing $Y(T+d)=0$ by the specification of
  \textsc{Mem}---a contradiction.)
  \item Let $t_0\in [t^*,t^*+d]$ be maximal with the property that
  $Y(t)=0$ for all $t\in [t^*,t_0]$ ($t_0$ exists because $Y(t^*)=0$ and
  $\{t\in \R_0^+\,|\,Y(t)=0\}$ is closed).
  \item $X(t)=0\vee Y(t)=1$ for every $t\in (t_0-T,t_0)$ (specification of
  \textsc{Mem}).
  \item (a) $t_0<t^*+d$: Then, $Y(t_0^+)=\lim_{\varepsilon\to 0+}Y(t_0+\varepsilon)=1$ 
  by maximality of $t_0$; hence $X(t_0)=1$
  by the specification of \textsc{Mem} and 3.\\
  (b) $t_0=t^*+d$: Then, the channel ensures $X(t_0)=\overline{Y}(t^*)=1$.
  \item $Y(t)=1$ for $t\in (t_0,t_0+T)$ (specification of \textsc{Mem}).
  \item $X(t)=\overline{Y}(t-d)=0$ for $t\in (t_0+d,t_0+T+d)\supseteq
  (t_0+T,t_0+T+d)$ (as $T\geq d$).
  \item $Y(t)=0$ for $t\in [t_0+T,t_0+T+d]$ (specification of \textsc{Mem}).
\end{compactenum}
We can now determine $Y$ for larger times inductively, showing for
$t_i:=t_0+i(T+d)$, $i\in \N$, that $Y(t)=1$ for $t\in (t_i+d,t_i+T+d)$ and 
$Y(t)=0$ for $t\in [t_i+T,t_i+T+d]$. Hence, the execution is feasible for module
\textsc{osc} during $[t_0,\infty)$. As $t_0\leq t^*+d\leq T+2d$, the claim
follows.
\end{proof}
By contrast, choosing $d > T$ leads to a circuit that does not necessarily
  self-stabilize.

While the circuit in Figure~\ref{fig:osc_mem} is a self-stabilizing
  implementation of \textsc{osc}, it is not fault-tolerant.
A single permanent fault will stop it from operating correctly.
However, the principle of using ``forgetful'' memory to achieve
  self-stabilization of oscillatory circuits can
  be carried over to fault-tolerant \emph{distributed} oscillators like DARTS:
In \cite{dolev14,dolev14fatal}, we leveraged the approach in the design of fault-tolerant and
  self-stabilizing solutions to clock generation (and
  clock distribution~\cite{DFLPS13:HEX}).
FATAL$^+$, the proposed clock generation
  scheme, is essentially a distributed oscillator composed of $n\geq 3f+1$ clock
  generation nodes, which self-stabilizes in time $\BO(n)$ with probability
  $1-2^{-\Omega(n)}$ even in the presence of up to $f$ Byzantine
  faults~\cite{dolev14fatal}.

\section{Module Composition Artefacts: The Weird Module}
\label{sec:weirdmodule}

In this section, we will show that module composition in our framework
  sometimes leads to surprising effects.
As a consequence, one has to be careful when composing
  innocently looking modules that hide
  their true complexity behind deceptively simple specifications.  

Like in \cref{sec:transientfaults}, we will restrict signals and 
  feasible executions in module specifications from $\R$ to arbitrary
  subintervals $I=[t^-,t^+]\subseteq \R$.
To make such interval-restrictions explicit, we will sometimes write
  $\sigma^I$, $\In^I$, $\Out^I$, $\E^I$ etc.
Note that such intervals $I$ could also be open $(t^-,t^+)$, closed
  $[t^-,t^+]$, or half-open, but are always contiguous. 

A sequence of signals $(\sigma_i)_{i\in C}$ defined on intervals 
  $I_1\subseteq I_2\subseteq\ldots$ with $I_i \subseteq I_{i+1}$ for all 
  $i\in C$, for $C=\{1,\ldots,n\}$ or $C=\N$, is called a \emph{covering} 
  of a signal $\sigma$ defined on $I=\bigcup_{i\in C}I_i$ if, 
  for all $i \in C$, $\sigma_i=\sigma^{I_i}$.
Clearly, any sequence of signals $(\sigma_i)_{i\in C}$ on
  $I_1\subseteq I_2\subseteq\ldots$ 
  with the property that $\sigma_i=\sigma_{i+1}^{I_i}$ for
  all $i \in C$ defines a unique $\sigma$ on $I=\bigcup_{i\in C}I_i$ 
  such that $(\sigma_i)_{i\in C}$ is a covering of $\sigma$. 
For $C=\N$, we can hence set $\lim_{i\to \infty}\sigma_i = \sigma$,
  where $\sigma$ is defined on $\lim_{i\to \infty} I^i = I$.
  These definitions and results naturally carry over to interval-restricted
  executions, i.e., pairs of sets of input and output signals of a module.

\begin{definition}[Limit-closure]\label{def:limitclosure}
Module $M$ is \emph{limit-closed} iff, for every covering
$(\E_i)_{i\in \N}$ consisting of interval-restricted
executions $E_i$ in the set $\ME_M$ of all interval-restricted
executions of $M$, it holds that 
  $\lim_{i\to \infty}\E_i\in \ME_M$.
\end{definition}

Not every module is limit-closed, as the following example demonstrates.

\begin{example}
Consider the module specification $\WM$, subsequently called 
  the\linebreak \emph{weird module}: It has no 
  inputs and only a single output, which is required to 
  switch from $0$ to $1$ within finite time and have no other
  transitions.
\end{example}
The $\WM$ can be seen as an archetypal asynchronous module, 
  as the transition must occur in finite time, but there is no known bound on the time
  until this happens. 

For every $i\in \N$, the execution $\E_i$ defined 
  on $[-i,i]$ with the output signal being constant $0$ is feasible 
  for $\WM$, as it can be extended to some execution on $\R$ where the 
  transition to $1$ occurs, e.g., at time $i+1$;
  it is thus an extendible module according to Definition~\ref{def:extendible}.
  However, the limit of $\E=\lim_{i\to\infty}\E_i$ is the unique
  execution on $\R$ with
  output constant $0$, which is infeasible for $\WM$.
According to \cref{def:limitclosure}, the specification of $\WM$ is hence extendable
  but not limit-closed.
Conversely, limit-closure does not necessarily
  imply extendibility either, as the latter requires that \emph{every} execution 
  defined on some interval $I$ can be extended to an execution on $\R$; 
  limit-closure guarantees this only for interval-restricted executions
  that are part of coverings.

\begin{definition}[Finite-delay \& bounded-delay module]\label{def:FDmodules}
Module $M$ has \emph{finite delay} (FD), iff every infeasible execution $\E_M\not\in \ME_M$ has a
  finite infeasible restriction, i.e.,
%\begin{align*}
$\left( \E_M \not\in \ME_M \right) \Rightarrow
\left( \exists \mbox{ finite $I \subset \R$}: \E_M^I \not\in \ME_M \right)$.
%\end{align*}
An FD module $M$ is a \emph{bounded-delay module} (BD), if $I$ 
  may not depend on the particular $\E_M$ in the FD definition. 
Finally, $M$ is a \emph{bounded delay module with delay bound $B\in \R^+_0$},
  if it is BD and $|I| \leq B$.
\end{definition}

Recall that if a module is correct in an execution during an interval $I$
  it is always also correct within a subinterval of $I$ in the same execution.
On the other hand, the other implication
$(\exists \mbox{ finite $I \subset \R$}: \E_M^I
\not\in \ME_M) \Rightarrow (\E_M \not\in \ME_M)$ always holds, by our
definition of a restriction.
Thus, for FD modules, it holds that
%\begin{align*}
$\left( \E_M \not\in \ME_M \right) \Leftrightarrow
\left( \exists \mbox{ finite $I \subset \R$}: \E_M^I \not\in \ME_M \right)$. 
%\end{align*}

According to \cref{def:FDmodules}, $\WM$ is not a finite-delay 
  module, as any finite restriction of the infeasible all-zero trace on $\R$ is
  feasible.
More generally, we have the following lemma:

\begin{lemma}\label{lemma:limit-closed}
A module is limit-closed iff it has finite delay.
\end{lemma}
\begin{proof}
Suppose $M$ is limit-closed. Given an arbitrary $\E\notin \ME_M$, defined on 
$I\subseteq \R$, consider the covering $\E_i=\E^{I\cap [-i,i]}$, $i\in \N$. Then, either there is some $i$ so that
$\E_i\notin \ME_M$, or we reach the contradiction that $\E=\lim_{i\to
\infty}\E_i\in \ME_M$ as $M$ is limit-closed.

Conversely, suppose that $M$ is a finite delay module. 
Consider an arbitrary infinite covering $\{\E_i\,|\,\E_i\in
\ME_M\}_{i\in \N}$, and denote by $\E$ its limit. If $\E\notin \ME_M$, then by
\cref{def:FDmodules} there is a finite $\E^I \notin \ME_M$ that is a restriction of $\E$. As
$\{\E_i\}_{i\in \N}$ is a covering, for sufficiently large $i$, it holds that
the interval $I$ on which $\E^I$ is defined is contained in the interval on which
$\E_i$ is defined. Hence, $\E^I \notin \ME_M$ is a restriction of $\E_i\in \ME_M$, which is a contradiction to the fact that $\E^I$ must be feasible.
\end{proof}

At that point, the question arises whether and when the composition of
  modules preserves bounded resp.\ finite delays. 
The following \cref{cor:presBD} shows that 
  this is the case for feedback-free compositions of BD modules:

\begin{corollary}[Preservation of BD]\label{cor:presBD}
Suppose compound module $M$ is feed\-back-free with  circuit graph $G_M$
  and each of its submodules $S\in {\cal S}_M$ is FD.
Then, $M$ is FD.
Moreover, if $S\in {\cal S}_M$ is BD with delay bound $B_S$, then ${\cal S}_M$ is
  BD with delay bound 
\begin{equation*}
B = \max_{\substack{(S_1,\ldots,S_k)\\\text{path in }G_M}}\left\{\sum_{i=1}^k
B_S\right\}.
\end{equation*}
\end{corollary}

\noindent\textit{BD is not preserved in arbitrary compound modules.}
Unfortunately, the above corollary does not hold if feedback-loops are
  allowed.
A compound module made up of BD submodules in a feedback-loop need not be
  BD, and sometimes not even FD. 

As an example of the former, consider the \emph{eventual short-pulse filter} 
  ($\eSPF$) introduced in \cite{FNS16:ToC}, which has a single input and a 
  single output port, initially 0. 
Given a single pulse of duration $\Delta>0$ at time $0$ at the input,
  there is a time $T=T(\Delta)$ with 
  $\lim_{\Delta\to 0} T = \infty$ such that the output $o(t)=1$
  for all $t\geq T$.
Yet, the execution where the output never settles to 1
  is not feasible (unless there is no input pulse).
$\eSPF$ can be implemented as a compound module consisting of a
  two-input zero-time \textsc{Or} gate and a two
  pure-delay channels (with delay $1$ and $\sqrt{2}$, respectively) 
  in a feedback-loop. 
By adding an inertial delay channel \cite{FNNS19:TCAD} to the output
  of eSPF, which suppresses all pulses with duration less than~1 (and is
  hence also a BD module\footnote{Since we can consider the inertial
  delay channel to be a basic module here, we need not care about
  implementability at that point.}), we
  obtain a module $\eSPF'$ that generates exactly one transition from
  0 to 1 at the output. 
Module $eSPF'$ is FD, as a finite interval
  $I=[0,T]$ that guarantees $\E^I \in \ME_{\eSPF'}$ can be computed from 
  the known $\Delta$ in every given execution $\E$.
However, $\eSPF'$ is not BD, albeit all its submodules are BD.
Consequently, for compound modules that are not feedback-free, 
  \cref{cor:presBD} need not hold.

\noindent\textit{Unexpected properties of the $\WM$ module.}
While the results on BD align with our intuition, similar
  properties do not hold for FD modules.
To show this, let us add another BD basic submodule that acts 
  as a random generator (which is of course also BD) for
  generating an input pulse of duration $\Delta>0$ to $\eSPF$ (now considered
  a basic module).
We obtain a feedback-free compound module implementation of $\WM$,
  which is not even FD!
Consequently, and surprisingly, one cannot generalize
  \cref{cor:presBD} to the preservation of FD:
Feedback-free compound  modules composed from FD submodules are \emph{not} always FD: 
  $\eSPF'$ is FD, and the random
  generator is even BD, yet the resulting $\WM$ is not FD.

The problem can be traced back to the fact that
  compound modules hide internal ports (the input
  port of $\eSPF$ fed by the random generator in our $\WM$ compound module), 
  which does no longer allow to identify appropriate infeasible finite 
  executions in infinite executions according to \cref{def:FDmodules}.
Our modeling framework allows to completely abstract away this 
  important submodule-internal information, which in turn creates 
  this artefact.
This intuitively suggests that one should not entirely discard 
  the internal structure of a compound module, but rather simulate
  a ``glass box view'' of a submodule as advocated in \cite{BS01:focus}
  by exposing important submodule-signals when composing modules.
A formal understanding of these problems is open, however.

\section{Outlook}
\label{sec:outlook}

We discussed an alternative to the classical state-based modeling and analysis
  approach.
In particular, we reviewed the state-oblivious modeling
  and analysis framework introduced in \cite{dolev14}, and argued its
  utility by means of some examples, in particular, a self-stabilizing
  oscillator.
We also showed that it may create some subtle artefacts when 
  composing modules, which need careful consideration and possibly mitigation.

While we believe that the modeling framework discussed in this article 
  is a sound basis for the formal study of digital circuits and even biological
  systems, it currently lacks several important features that are left open for
  further research:

  The first one is the choice of a formal language for describing signals
  and module specifications. Whereas our simple signals could of course
  be described within a first-order theory on $\R$, it is not clear whether
  this is the most appropriate formalism for concisely expressing the most
  relevant properties of interest. Moreover, module specifications often
  require (all-)quantification over signals, which suggests the need for
  a second-order theory.

  A somewhat related open problem is the
  definition of a proper notion of simulation equivalence for modules
  with \emph{different} interfaces, and simulation-type
  proof techniques similar to the ones known for both untimed \cite{LV95:1} and timed 
  \cite{LV96:2} distributed systems. Unfortunately, the state-obliviousness
  and the unconventional domain $\R$ of our framework does not allow one to just take over state-based simulation techniques.

Another open issue is the explicit handling of metastability, which can currently
  only be expressed by mapping a metastable state to a (high-frequency) pulse
  train.
An obvious alternative is to use a three-valued logic, also
  providing a dedicated metastable state $M$, as advocated in \cite{FFL18:ToC}.
While this extension appears relatively straightforward at the specification level,
  it should also be accompanied by ways of specifying metastable upsets 
  and metastability propagation.

{%\small
\bibliographystyle{splncs04}
\bibliography{../pulse}
}

\end{document}